\documentclass[12pt, reqno]{amsart}
\usepackage{geometry}
\usepackage{amsthm}
\usepackage{amssymb}
\newtheorem{definition}{Definition}[section]
\newtheorem{theorem}{Theorem}[section]
\newtheorem{lemma}{Lemma}[section]
\usepackage{enumerate}
\usepackage{graphicx}
\graphicspath{{./images/}}
\usepackage[T1]{fontenc}
\usepackage[utf8]{inputenc}
\numberwithin{equation}{section}
\usepackage[nomarkers]{endfloat}

\let\emptyset\varnothing

\title{Universal Computation is `Almost Surely' Chaotic }
\author{Nabarun Mondal}
\address{D.E.Shaw \& Co. India, Hyderabad }
\email{mondal@deshaw.com}
\thanks{Nabarun Mondal : 
Dedicated to my late professor Dr. Prashanta Kumar Nandi. \\ 
Dedicated to my parents. \\
Big thanks to:- Abhishek Chanda, Samira Sultan and Shweta Bansal : You all have been constant support. \\
In Memory of : Dhrubajyoti Ghosh. Dear Dhru, rest in peace.
}

\author{Partha P. Ghosh}
\address{Microsoft India, Hyderabad }
\email{parthag@microsoft.com}
\thanks{ Partha. P. Ghosh : 
Dedicated to my parents and family without their presence we are nothing. \\
}

\subjclass[2010]{Primary 03D10 ; Secondary 65P20,68Q05,68Q87}  

\begin{document}

\keywords{
Turing Machines ; Universal Computation ; Chaos ; Metric Space ; Measure Space ; With Probability One ; Aleph Numbers  
}

\begin{abstract}
Fixed point iterations are known to generate chaos, for some values in their parameter range.
It is an established fact that Turing Machines are fixed point iterations. 
However, as these Machines operate in integer space, the standard notions of a chaotic system
is not readily applicable for them. Changing the state space of  Turing Machines 
from integer to rational space, the condition for chaotic dynamics can be suitably established, 
as presented in the current paper. 
Further it is deduced that, given random input, computation performed 
by a Universal Turing Machine would be `\emph{almost surely}' chaotic. 

\end{abstract}

\maketitle

\textbf{ According to the Church-Turing thesis, there is an abstract computing device, 
called Universal Turing Machine(UTM), which can simulate any effective computation. 
It well known that a UTM is a type of iterative map or dynamical system. Dynamic Systems
can exhibit chaotic behavior, and in this paper condition of chaos in a UTM has been derived. 
It has been formally proven that, given a random computation to simulate on any UTM, 
the resulting dynamics will be `almost surely' (with probability 1) chaotic. }

\begin{section}{Introduction}\label{intro}

Alan Turing's abstract computing device, \emph{Turing Machine}, is in effect a \emph{fixed point iteration}.
While computation is believed to be of predictable nature,  
generic fixed point iterations do not always behave similarly. 
Some iterations do indeed converge to a fixed point,while others diverge to infinity, 
yet a few oscillate between multiple points. But in some cases, 
the iterations neither converge, nor diverge, and do not even oscillate between points. 
This \emph{deterministically random} behavior (Section \ref{fpi-chaos}) is informally termed as \emph{chaos}.

Many theoreticians believe absence of direct connection between computation and chaos.
However, many dynamic systems which are capable of simulating 
\emph{Universal Computation} for some configurations, 
also exhibits chaotic behavior for some other configuration \cite{anks}.
Most prominent example of this behavior is cellular automata \cite{anks}. 

In Section (\ref{comp-itr}) a brief discussion establishes Turing Machines as fixed point iterations,
and lays the foundation of our work. Section (\ref{chaos-computation}) demonstrates the concept of Chaos 
in the Universal Turing Machines. Further, it follows (Section  \ref{prob-chaos-tm}) that 
\emph{Universal Computation, `almost surely' , will be chaotic}.        

As a closure (Section \ref{implication}) three philosophical implications of this finding are discussed.

\end{section}

\begin{section}{Fixed Point Iteration and Chaos}\label{fpi-chaos}

Let us define a function, $f:X \to X$ , where its domain and range are  the same.
Starting with an initial input `$x_0$' to the function, let `$f$' produce the output `$x_1$', such that,
$$
x_1 = f(x_0) .
$$ 

The output of the function in each step can be taken as the input to next step and hence, 
\begin{equation}\label{fpi}
x_{n} = f(x_{n-1}) \;  ; \; n \ge 1 
\end{equation}

Equation \eqref{fpi} is called  a \emph{fixed point iteration}, 
and is the simplest example of a \emph{dynamic system} \cite{ap}\cite{mbgs}.

Iteration \eqref{fpi} may result in $x_n$ being converged to a limit point, diverge to infinity, 
or even exhibit chaos. Banach Fixed Point theorem \ref{bfpt}, 
presented in Appendix-\ref{ap_1}, details the condition under which the iteration would converge.

For a fixed point iteration of type \eqref{fpi} , defined over some interval $f:I \to I$,
the inverse of the Banach theorem \ref{bfpt} has an interesting implication.
 \emph {Iterates will not converge to one point} 
if at least one of the following conditions is true.

\begin{enumerate}\label{non-convergence}
\item {Function `$f$' has \emph { no solution } in $I$ : 
$ \not \exists x^*  \in I \; s.t. \; f(x^*) = x^* $ . }
\item {The interval $I$ is  \emph{not complete} . }
\item {\emph { `$f$' is not a contraction mapping } satisfying \eqref{cont_map}.}
\end{enumerate}

The iterate \eqref{fpi}, then, will not converge, and will either diverge or 
\emph{may even} run in a chaotic manner.

However, \emph{Chaos} is a tricky term to define. 
It is much easier to list down properties that a  \emph{chaotic system} must possess, 
than giving a precise definition of chaos.
Appendix \ref{ap_1} has the related definitions.

\begin{definition}\label{chaos}
\textbf{Characteristics of Chaotic Dynamics.}
 
\begin{enumerate}
\item{Sensitivity to the initial condition of the system (where the neighborhood of the initial point can quickly lead the system into very different final states.).}
\item{Having a dense (definition \ref{dense-set}) collection of points with periodic orbit (definition \ref{orbit} ). }
\item{Topologically mixing (definition \ref{top-trans}). }
\end{enumerate}
\end{definition}

\emph{(3)} and \emph {(2)} in the definition \eqref{chaos} 
imply the \emph {sensitive dependence on initial conditions (1)}.

Some dynamical systems, like the one-dimensional logistic map (equation \eqref{logistic}) with $r=4.0$,  
are chaotic everywhere. However in many cases chaotic behavior is found 
only within a subset of phase space.

We present three functions exhibiting \emph{chaotic dynamics} in the following subsection.

\begin{subsection}{Chaos in the Fixed Point Iterations}

\begin{subsubsection}{The Chaotic Search for $i=\sqrt{-1}$ }
The Babylonian method employs an iterative map to calculate $\sqrt{S}$
with arbitrary precision.

\begin{equation}\label{sqroot_map}
x_{n+1} = \frac{1}{2} ( x_n + \frac{S}{x_n}) \; ; \; n \ge 0 \; ; \; x_n \in \mathbb{R} 
\end{equation}

When  $S<0$ , there is no $x^* \in \mathbb{R}$ that would satisfy ${x^*}^2 = S$.
The iteration can clearly not converge. It is discussed in \cite{csi} that the iteration eventually becomes chaotic.
\end{subsubsection}

\begin{subsubsection}{Chaos in the Logistic Map}
Logistic map is a very well known example of an iterative map
demonstrating chaotic dynamics :-

\begin{equation}\label{logistic}
x_{n+1} = r x_n (1 - x_n) \;  ; \; r \in [0,4] \; ; \; x \in (0,1)
\end{equation}
When $r=3$ , the system bifurcates into two fixed values. If $r$ is increased further, 
the bifurcations continue further, till reaching a chaotic state at $r=4$. 
There are numerous discussions in \cite{cfnfs} of this behavior.  

\end{subsubsection}

\begin{subsubsection}{Chaos in Aperiodic, Bound, Non Convergent Fixed Point Iterations}

We show here that such a function in fact exhibits chaotic dynamics, 
much like the Babylonian method to search for $\sqrt{-1}$.
Demonstrating \emph{chaos} in an iteration would require one to show that the system is topologically mixing, 
having a dense orbit.  
 
Lemma \ref{non-cauchy-bounded-seq} proves the topologically mixing property of the system.
The construction used in the lemma can be used to show that the orbit is dense.

\begin{lemma}\label{non-cauchy-bounded-seq}
\textbf{ Topological mixing  in an Aperiodic, Bounded, Non-Cauchy sequence.}

Let $\mathcal{N} = <x_n>$ be a  sequence with $x_n = f(x_{n-1})$ with $f:X \to X $ such that:-
\begin{enumerate}
\item{ $\forall n \; ; \; l \le x_n \le u $ , where $-\infty < l \le u < \infty$ are its bounds ;}
\item{ $x_i = x_j \implies i=j$ ; that is the sequence has a period of infinity ; }
\item{ $<x_n>$ is not a Cauchy Sequence ;}
\end{enumerate}
Then, there exists a set $\Psi  \subseteq X $  where `$f$' is topologically mixing.   

\end{lemma}
\begin{proof}[Proof of the lemma \ref{non-cauchy-bounded-seq}.]

Using lemma \ref{ncbs} (in Appendix \ref{ap_1})  we assert that, for $\mathcal{N} = <x_n>$ , 
there exists a countable set, $\mathcal{C}$, of Cauchy sequences, 
which are subsequences of $\mathcal{N}$. 
We assume $|\mathcal{C}| = n$.

Clearly every Cauchy sequence in $\mathcal{C}$, $\mathcal{C}_i$ , has a limit point $a_i$.  
If $a_i \not \in X$ , $X$ must be an incomplete metric space (for example $X \subseteq \mathbb{Q}$ ).
We can use arbitrarily close approximation of $a_i$ by virtue of density of rational numbers ($\mathbb{Q}$) in real ($\mathbb{R}$) set
( by definition \ref{dense-set} ) and name it $\mathbf{a_i}$. 
If $X$ is a complete metric space, we define $\mathbf{a_i} = a_i$.

The set of approximate limit points, so constructed, can be defined as follows.
$$
\mathcal{A}=  \{ \mathbf{a_i} \;|\; \mathbf{a_i} = lim(\mathcal{C}_i)   \}
$$

where each $\mathbf{a}_i$ is distinct. 

Cauchy subsequence, $\mathcal{C}_i$ , whose limit point is $\mathbf{a_i}$ 
eventually will be entering the interval $I_i = |x - \mathbf{a_i}| <\epsilon $, 
where $\epsilon$ is an arbitrarily small number. 
Hence, every single point  $x_i \in I_i$ , can be included in the Cauchy subsequence $\mathcal{C}_i$.

Now, as $\mathcal{C}_i$ is a Cauchy sequence, $\exists N_i$ such that,
$$
\forall i \; : \; n,m > N_i \implies |x_m - x_n| < \epsilon ;
$$ 

However, these sequences are `mixed' or `interleaved' with one another. 
Assume that $k_i$'th element of the sequence $\mathcal{N}$ 
is identified with the $N_i$'th element of the sequence, $\mathcal{C}_i$.

We can now generate the set $K = \{ k_1,k_2,k_3,...,k_n \}$. 

Obviously, there is a $K_M = MAX(K)$, where MAX is the maximum value of a set.
$K_M$ would vary depending upon the choice of `$\epsilon$'. Nevertheless,
there will always exist one $K_M$ for any choice of `$\epsilon$'.  

Based on a unique choice of $\epsilon$ , hence $K_M$, we have,
$$
\forall i \; : \; x_n,x_m  \in \mathcal{C}_i \; ; \; n ,m > K_M \implies |x_m - x_n | < \epsilon \; ;
$$ 

Now, we construct the sequence $\mathcal{N}^* = \{ x_i \in \mathcal{N} \; ; \; i \ge K_M \; \}$.

Let the set of ordered limit points ( $i<j \implies \mathbf{a_i} < \mathbf{a_j}$ )  be,
$$
\mathcal{A}^*=  \{ \mathbf{a_1},\mathbf{a_2},\mathbf{a_3},...,\mathbf{a_n} \}
$$
We create a set of `$n$' sequences,
$\mathbb{A} = \{ \mathcal{A}_1^*,\mathcal{A}_2^*,\mathcal{A}_3^*,...,\mathcal{A}_n^* \}$, 
as follows:-
\begin{enumerate}
\item{ Take $x^* \in \mathcal{N}^*$ , and remove it from $\mathcal{N}^*$. }
\item{ Find all $d_i(x^*) =|x^* - \mathbf{a_i}|$. Let the minimum be achieved at $i=k$ i.e 
$d_k(x^*)= MIN( d_i(x^*) )$, where MIN stands for the minimum value of a set.}
\item{ Assign $x^*$ to the $k$'th newly created sequence : $\mathcal{A}_k^*$. 
In case there are two sequences $(l,k, \; l<k)$ 
satisfying $MIN( d_i(x^*)$,assign $x^*$ to the one having lower index i.e. : $\mathcal{A}_l^*$.  
}

\item {Repeat the above steps (1,2,3) till $\mathcal{N}^*$ is exhausted.}

\end{enumerate}

By virtue of construction, we achieved `$n$' centroid clustering of the sequence  $\mathcal{N}^*$.
The same centroids are also present in $\mathcal{N}$.

Every sequence $\mathcal{A}_i^*$, has $l_i = MIN(\mathcal{A}_i^*)$ and $u_i = MAX(\mathcal{A}_i^*)$.

Using $l_i$ and $u_i$ , center $\mathbf{a_i}$ to create intervals as follows:-
\begin{equation}\label{disj-int_1}
\mathcal{I}_i = [ l_i ,  u_i ] 
\end{equation}

We now show that the neighboring intervals $\mathcal{I}_i$ and $\mathcal{I}_{i+1}$ are disjoint.

The construction of the intervals  implies the following.
$$
x \in \mathcal{I}_i \implies |x - \mathbf{a_i}| \le |x - \mathbf{a_j}| 
$$ 
\begin{equation}\label{disj_1}
u_i - \mathbf{a_i}  \le  \mathbf{a_{i+1}} - u_i \implies u_i \le  \frac{1}{2}(\mathbf{a_i} + \mathbf{a_{i+1}}) 
\end{equation}
\begin{equation}\label{disj_2}
\mathbf{a_{i+1}} - l_{i+1} \le   l_{i+1} - \mathbf{a_i}  \implies l_{i+1} \ge  \frac{1}{2}(\mathbf{a_i} + \mathbf{a_{i+1}})  
\end{equation}
 
Combining relations  \eqref{disj_1} and \eqref{disj_2} we find, $u_i \le l_{i+1}$.
However, as asserted, no elements repeat in $\mathcal{N}$ and subsequently in $\mathcal{N}^*$. 
Hence  $u_i \ne l_{i+1}$ , and $u_i < l_{i+1}$.
Therefore, all the $\mathcal{I}_i \in \{ \mathcal{I}_k\}$ are disjoint,
\begin{equation}\label{disj-int_2}
i \ne j \implies \mathcal{I}_i \cap \mathcal{I}_j = \emptyset 
\end{equation}
Now, we construct a set $\Psi$ as the union of all the `$n$' disjoint intervals 
corresponding to `$n$' sequences, $\mathbb{A} = \{ \mathcal{A}_i^* \}$.
\begin{equation}\label{psii}
\Psi = \bigcup _{i=1}^n \mathcal{I}_i
\end{equation}

Using  equation \eqref{disj-int_2} and \eqref{psii}, 
we can say that the set, $\mathcal{I} = \{  \mathcal{I}_k \}$, 
defines a basis for topology (definition \ref{basis} ) in $\Psi$.

As the function is aperiodic, it is clear that the function `$f$' mixes all the $\mathcal{I}_i$ intervals, 
so that $\exists n > 0 $ such that,
\begin{equation}\label{mixing-nc-int}
\forall i, j \; \;  f^n(\mathcal{I}_i) \cap \mathcal{I}_j \ne \emptyset 
\end{equation}

Now, consider two non empty open sets : $A,B \subseteq \Psi$ , 
such that $A\cap B = \emptyset$. 
That would mean, $\exists \mathcal{I}_k \; , \mathcal{I}_r$ such that,
$$
B \cap \mathcal{I}_k \ne \emptyset \; ; \; A \cap \mathcal{I}_k  = \emptyset 
$$  
and
$$
B \cap \mathcal{I}_r = \emptyset \; ; \; A \cap \mathcal{I}_r  \ne \emptyset 
$$  
Using equation \eqref{mixing-nc-int}, we can say,
$$
f^n(\mathcal{I}_r) \cap \mathcal{I}_k \ne \emptyset
$$
This will readily imply,
$$
f^n(A) \cap B \ne \emptyset
$$
which proves the assertion about mixing.
\end{proof}

\begin{theorem}\label{ncabmc}
\textbf{Any iterative map whose orbit is Bound, Aperiodic and Non-Cauchy, is chaotic.}

Let $f:X \to X$ be a map where $X$ is an interval. 
Let the orbit, $\mathcal{O}$, of $f$ be bound, aperiodic and Non-Cauchy.
Then, $f$ is chaotic. 
\end{theorem}
\begin{proof}[Proof of the theorem \ref{ncabmc} ]

Invoking Lemma \ref{non-cauchy-bounded-seq} we can say that the function is topologically mixing
within  $\Psi \subseteq X$. $\Psi$ is defined as in \eqref{psii}.

We now show that the orbit of the function $f$ is dense in $\Psi$. 
To do that, we take any arbitrary point $x \in \Psi$. 
By definition, $\Psi$ is an union of non overlapping intervals 
$\mathcal{I}_i$ \eqref{disj-int_1}. Clearly then,
$$
\forall \; i \ne j \; ; \;  x \in \mathcal{I}_i \implies x \not \in \mathcal{I}_j    
$$  
By virtue of construction (Lemma \ref{non-cauchy-bounded-seq}) of $\mathcal{I}_i$ :-
$$
x_1,x_2 \in \mathcal{I}_i \implies |x_1 - x_2| <  \epsilon .
$$  
This construction would imply that $\exists x^* \in \mathcal{I}_i$ such that:-
$$
|x - x^* | < \epsilon 
$$
By definition of dense set (definition \ref{dense-set}) then, the orbit $\mathcal{O}$ of $f$ is dense within $\Psi$.

It is evident that `$f$' exhibits both dense orbit, and is topological mixing. 
By definition \ref{chaos} of Chaotic system, then, $f$ is chaotic. 
\emph {However, its orbit $\mathcal{O}$ is not periodic.} 
\end{proof}

Computer simulation of the following iterative function,
\begin{equation}\label{tan_x}
x_n = tan(x_{n-1})
\end{equation}

displays above behavior as experiments have shown. It is bound, 
and its singular points are irrational in nature, 
which can not be approximated well enough in a digital computer (except $x=0$ ). 
The aperiodic nature is \emph{conjectured}, 
as even after a billion iterations there was no repeat experimentally recorded. 
All the three distinct features of chaos can be seen as this function iterates. 

\end{subsubsection}
\end{subsection}
\end{section}

\begin{section}{Computation and Iterative Maps}\label{comp-itr}

A Turing Machine (Appendix \ref{ap_2}), after performing a computation typically writes back to the tape. 
Every write in the tape can be considered a step, and the state of the tape
(the ordered sequence of symbol on the tape) can be denoted by $T$. 
By definition, thus, at every step, a state transition takes place based on the tape sequence.

Let the state of the tape at step `$n$' be designated as `$T_n$'. 
At the `$n+1$'th step, the state of the tape changes to `$T_{n+1}$'. 

Hence, a specific Turing Machine, `$M$', while computing, 
changes the state of the tape from $T_{n-1} \to T_n$ , and that is the \emph{computation} 
that took place in the $n$'th step.

Hence, a computation can be written as:-
\begin{equation}\label{t-iterate}
T_{n} = M(T_{n-1})
\end{equation}

Comparing equation \eqref{fpi} with equation \eqref{t-iterate}, we can suggest that Turing Machines 
are a very specific type of iterated maps. This has been discussed in \cite{bc} , 
and a formal proof based on neural network is established in \cite{hh}. 
The following theorem is due to Hy\"{o}tyniemi.

\begin{theorem}\label{turing-iterate}
\textbf{ (Hy\"{o}tyniemi) }

Turing Machines  can be written as  non linear discrete time system, of the form:-
$$
x_{k+1} = g(x_k) \; ; \; k \ge 0
$$
\end{theorem}

This theorem obviously leads to the following lemma : 

\begin{lemma}\label{utm-iterate}
\textbf{Universal Turing Machines are Iterative Maps.}
\end{lemma}
\begin{proof}[Proof of the Lemma \ref{utm-iterate}]
Every Turing Machine is an iterative map. 
A Universal Turing machine, being a Turing machine is one. 
\end{proof}

\begin{subsection}{Phase Space of Turing Machine}
According to equation \eqref{t-iterate} , the tape of the Turing Machine is the iterate variable,
with symbols from $\Gamma$.
But to be able to represent the phase(iterate variable) of a Turing machine, 
the tape needs to be converted into a number.  
Hence, we need to discover a way to map every $T_n$ into a number.
This conversion is established using \emph{G\"{o}delization} techniques.

\begin{definition}\label{god}
\textbf{G\"{o}delization (G\"{o}del).}

Any string from an alphabet set $\Gamma$ can be represented as an integer in base `$b$' with $b = |\Gamma|$. 
To achieve this, create a one-one and onto G\"{o}del map $g : \Sigma \to D_b$ , where,
$$
D_b = \{ 0 , 1, 2, ... ,b-1 \} .
$$ 
G\"{o}delization or $\mathbb{G} : \Sigma^+ \to \mathbb{Z_+}$ then, is defined as follows:
 
A string of the form $w = w_{n-1}w_{n-2}...w_1w_0$ , with  $w_i \in \Gamma$ ,
can be mapped to an integer $I_w = \mathbb{G}(w)$ as follows:
$$
I_w = \mathbb{G}(w) = \sum\limits_{k=0}^{n-1} g(w_k) b^{k}
$$
\end{definition}
The common decimal system is a typical example of G\"{o}delization of symbols from $\{ 0,1,..,8,9\}$.
The binary system represents G\"{o}delization of symbols from  $\{ 0,1\} $.
As a far fetched example, any string from  the whole english alphabet, can be written as a base 26 integers!
 
\begin{definition}\label{rat}
\textbf{Rationalization.}

Any string `$w$' of length `n' ($|w|=n$) , created from an alphabet set $\Gamma$,
can be represented as a rational number $x \in \mathbb{Q}$.
We define the rationalization, $\rho$ , in terms of G\"{o}delization (definition \ref{god}) as follows :
$$
x = \rho(w) = \frac{\mathbb{G}(w)}{ b^n } = \mathbb{G}(w) b^{-n} = 0.w_{n-1}w_{n-2}...w_0
$$
By definition, $x \in [0,1]$.
\end{definition} 

The definition \eqref{rat} of rationalization effectively means that 
we are implicitly adding a decimal point to the left of the representation
to convert it to a floating point (rational) number in base `b'.

Based upon the definitions of \eqref{god} and \eqref{rat}, we now define the phase of a Turing Machine as follows :

\begin{definition}\label{phase}
\textbf{Phase of a Turing Machine.}

The phase of a Turing Machine having right infinite tape, 
at step `$n$' with respect to a particular G\"{o}del map `$g$', 
is the rationalization (definition \ref{rat} ) of it's tape state, $T_n$ i.e. ,
$$
x_n = \rho(T_n) .
$$

By definition, then, the phase is always bound 
inside the interval $[0,1]$ . 
\end{definition}

\begin{definition}\label{phase-space}
\textbf{Phase Space of a Turing Machine.}

From the earlier definition of phase, the phase space of a Turing Machine is a metric space, 
defined as below :

$$
X_t = \{x| x\in \mathbb{Q} \; ; \; 0 \le x \le 1 \}
$$
The associated distance function, `$d(x,y)$' , of the above phase space  
is defined as $d(x,y)= |x-y|$.It induces an incomplete metric space.  
\end{definition}

\begin{definition}\label{tm-dynamic}
\textbf{Turing Machine is a discrete time Dynamic System.}

A Turing Machine can be thought of as an iterative map, $M : X_t \to X_t $ .
\end{definition} 

\end{subsection}

\end{section}

\begin{section}{Chaos And The Turing Machines}\label{chaos-computation}

Undecidable problems in Computation are more interesting from the point of view of chaos,
as a decider would surely halt. The end result would not be chaotic.

The famous ``Halting Problem'' (Appendix \ref{ap_2}) 
is an example of one such undecidable problem which was the first of such to be discovered.

\begin{subsection}{Implication of the Halting Theorem}

We define a debugger machine which contains a Universal Turing Machine 
(definition \ref{UTM}) as ``sub''-machine within.
It is capable of executing the instructions of the ``simulated'' Turing Machine one instruction at a time.
This machine is also capable of reading the whole tape of the Universal Turing Machine (till the blank symbol), 
rationalize (definition \ref{rat}) it, and put that number into a table.

By the definition of rationalization (definition \ref{rat}), 
it just involves copying the current content of the tape, end to end. 
For one sided tape bounded at left,
it copies from the start of the tape to the first blank symbol.
The real number `$x_n$' is implicit from the content of the tape.
The decimal point is \emph{assumed to be in the left side of the tape.}

\begin{definition}\label{debugger}
\textbf{Debugger Machine.}

A ``Debugger Machine'' is a 3-tuple of $\Delta_M = (H_{table},M_{utm},S)$ where,
\begin{enumerate}
\item{ `$M_{utm}$'  is a Universal Turing Machine. It will be simulating a Turing Machine with some input, where output of the simulated turing machine at any step is with the 
left bounded, right infinite tape `$T$'. }
\item{ `$H_{table}$' is a tabular sequence $<h_n>$  where rationalized `$T$' is stored at every step. 
$$
H= <h_n> \; ; \; h_n = \rho(T_n) .
$$ 
Clearly, $H \subseteq X_t$ .
}
\item{$S = \{ Y, N\}$ is a set of Symbols to be output for a simulation. 
It will output ``Y'' where system halts, else it will output ``N''.}
\end{enumerate}  

\end{definition}
The machine would operate as follows :-
\begin{enumerate}

\item{ The tape `$T$' of $M_{utm}$ is populated first with the description of a Turing Machine `$M$', and the input to simulate `$M$' with. }
\item{ The machine $\Delta_M$ would next rationalize the tape `$T$' , and store it in table $H = \{ \rho(T)\}$.}
\item{ The machine $\Delta_M$ now, will use $M_{utm}$ to execute  `$\delta$' transitions 
unless a transition leads to writing to the tape. }
\item{ Now, $\Delta_M$ will rationalize the current tape, $\rho(T_n)$ , and search it from `$H$'. }
\item{ If $\rho(T_n) \in H $, the machine `$M$' encountered a loop. Halt execution, and report a loop by outputting ``N'' . }
\item{ If $\rho(T_n) \not \in H $ , then, insert $\rho(T_n)$  in `$H$', and resume from (3). }
\item{ If the $M_{utm}$ has halted, output ``Y'', for success. }

\end{enumerate}
One example of such computing would be computing the square root with Babylonian method 
using equation \eqref{sqroot_map} . The Tape `$T$' then would have the `$x_n$'.
That is, the tape is the scratch pad for the Universal Turing Machine, it won't 
contain the \emph{encoding} of the embedded Turing Machine.

\end{subsection}

\begin{subsection}{Debugger and the Halting Problem}
The debugger machine (definition \ref{debugger}) is nothing but a \emph{ would be decider} Turing Machine. 
By the Halting Theorem (definition \ref{halting-theorem}) $\Delta_M$ is guaranteed not to halt at every input for the tape `$T$'.

In case $\Delta_M$ does not halt, the generated sequence `$H$' needs to be carefully analyzed.

Clearly, $H$ at this stage $\{ h_0, h_1, h_2 , ... , h_n \}$ is a bounded sequence, as 
$$
0 \le h_n \le 1 \; ; \; n \in [0,\infty)
$$
At this stage there are only three possibilities with the sequence $H = < h_n>$ :-

\begin{enumerate}

\item { terminates finally and becomes a finite sequence, due to 
	\begin{enumerate}
	\item { either debugger machine halting, }
    \item { or  nothing further is written on the tape, ever after.}
    \end{enumerate}
} 

\item {  becomes infinite Cauchy Sequence.}
\item {  becomes infinite but remains a Non-Cauchy sequence. }

\end{enumerate}

This 3rd behavior establishes the next theorem:-

\begin{theorem}\label{turing-chaos}
\textbf{Chaos in Computation (Debugger Machine).}

A non halting $\Delta_M$ , exhibits chaotic dynamics when $<h_n>$ is non terminating, and is not a Cauchy Sequence. 
\end{theorem}
\begin{proof}[Proof of the Theorem \ref{turing-chaos}.]
We note that the Turing machine $\Delta_M$ is a map:-
$$
\Delta_M : X_t \to X_t ,
$$  
and under the condition of the theorem
the orbit $<h_n>$ is aperiodic, bound and Non-Cauchy. Furthermore, $X_t=[0,1]$ is an interval.  
Therefore, we can invoke theorem \ref{ncabmc} to conclude that:-

$\Delta_M$ would exhibit chaotic dynamics.

\end{proof}

It is known that Universal Computation can be Devaney-chaotic \cite{dkb}.
This theorem reiterates the same from the perspective of Turing Machine.

The theorem \ref{turing-chaos} has significant practical consequences.
By tracing the sequence  $<h_n>$  for a long time, we can check 
if it indeed has been behaving like a convergent (Cauchy)sequence. 
If that is so, we can heuristically take a decision on halt the machine, 
assuming that the system would reach the limit point of the underlying Cauchy sequence.

However, clearly that is a probabilistic approximation.There is 
no guarantee that the sequence will \emph{remain Cauchy} after the observation.

\begin{theorem}\label{equiv-classes-h}
\textbf{Equivalent Classes on the Sequence $<h_n>$.}

The $<h_n>$ sequences can be categorized into three mutually exclusive set of sequences.
\begin{enumerate}
\item{Finite Sequences  $\mathcal{F}$ : 
These halt the debugger machine or do not write onto the tape
anything after finite time. }
\item{Infinite Cauchy Sequences,$\mathcal{C}$,each converging towards 
their limit points in $\mathbb{R}$.}
\item{Infinite Non-Cauchy sequences , $\mathcal{N}$ , which are chaotic sequences. }
\end{enumerate}

\end{theorem}
\begin{proof}[Proof of theorem \ref{equiv-classes-h}]
The proof is trivial using theorem \ref{turing-chaos}. 
\end{proof}

\end{subsection}

\end{section}

\begin{section}{Probability, Chaos and Turing Machines}\label{prob-chaos-tm}

In this section we  discuss the implication of the theorem \ref{turing-chaos} 
using very standard notion of `\emph{a.s}' within the realm of probability theory.

Appendix \ref{ap_3} has all the definitions needed to define the axioms of probability theory.

\begin{lemma}\label{debugger-h-cardinality}
\textbf{Cardinality of the sets $\mathcal{F}$ , $\mathcal{C}$ and $\mathcal{N}$. }

On the interval of rational numbers, $\mathbb{Q}= [0,1]$ , let :-
\begin{enumerate}

\item{$\mathcal{S}$ be the set of all sequences, finite or infinite, 
such that $x_i = x_j \implies i = j$ (no elements repeat ) ;}
\item{$\mathcal{F} \subset \mathcal{S}$ be the set of all  finite sequences (without repetition) ; }
\item{$\mathcal{C} \subset \mathcal{S}$ be the set of all Cauchy Sequences (infinite,without repetition) ; }
\item{$\mathcal{N} \subset \mathcal{S}$ be the set of all Non-Cauchy Sequences (infinite, without repetition).}

\end{enumerate}

Then $\mathcal{F}$ is null set with respect to $\mathcal{S}$, 
with $|\mathcal{S}| = |\mathcal{C}| = |\mathcal{N}|$.

\end{lemma}

\begin{proof}[Proof of lemma \ref{debugger-h-cardinality}]

We note that,
\begin{equation}\label{tp}
\mathcal{S} =  \mathcal{F} \cup \mathcal{C} \cup \mathcal{N}
\end{equation}

The set $\mathcal{S}$ is created by drawing one rational number at a time, without replacement from the set of rational numbers.
Clearly then, we have,
$$
|\mathcal{S}| = |\mathbb{Q}| \times |\mathbb{Q} - 1| \times | \mathbb{Q} - 2 | \times ...
$$ 

As $|\mathbb{Q}| = \aleph_0$, the above can be written as,
$$
|\mathcal{S}| = \aleph_0 (\aleph_0 - 1) ( \aleph_0 - 2 ) ...
$$ 
Using aleph arithmetic, we get,
$$
|\mathcal{S}| = \aleph_0 \aleph_0  \aleph_0  ... = \aleph_0 ^ {\aleph_0}
$$ 
Since 
$$
\aleph_0 ^ {\aleph_0} = \aleph_1
$$ 
we must have:-
\begin{equation}\label{card-s}
|\mathcal{S}|= \aleph_1 
\end{equation}

The set  $\mathcal{F}$  can be mapped one to one with the Natural Numbers Set.
It can be thought of as the set of numbers having finite decimal representation in base `$b$' by virtue of the rationalization \eqref{rat}. 
Then it is pretty obvious to arrive at equation \eqref{card-f}.  
\begin{equation}\label{card-f}
|\mathcal{F}| = \aleph_0
\end{equation}

Using equations \eqref{card-s} and \eqref{card-f} we can conclude that $\mathcal{F}$ is null set w.r.t. $\mathcal{S}$.

It is well known that the cardinality of the set $\mathcal{C}$ 
is the cardinality of continuum, that is:-
\begin{equation}\label{card-c}
|\mathcal{C}| = \aleph_1
\end{equation}

Using lemma \eqref{ncbs} we can say that every single Non-Cauchy Bounded Sequence 
comprises \emph{countable} Cauchy sequences. 

Consider the general class of Non Cauchy sequences, $\mathcal{N}(k)$, 
which are mixture of `$k$' distinct Cauchy sequences. Clearly, 
$$
\mathcal{C} = \mathcal{N}(1).
$$   
Then clearly:-
\begin{equation}\label{indexed-n-card}
 |\mathcal{N}(k)| =  {  | \mathcal{N}(1) |   \choose k  } M(k) = {  | \mathcal{C} |   \choose k  } M(k) = {  \aleph_1   \choose k  } M(k) = \aleph_1 M(k)
\end{equation}
where $M(k)$ is the all possible way `k' chosen infinite ( $\aleph_0$ ) sequences can be mixed.
However:- 
\begin{equation}\label{def-n-union}
\mathcal{N} = \bigcup_{k=2}^{\aleph_0} \mathcal{N}(k)   
\end{equation}
with disjoint $\mathcal{N}(k)$ :-
$$
\mathcal{N}(i) \cap \mathcal{N}(j) = \emptyset \; ; \; \forall i \ne j \; ; \; i,j \in \mathbb{Z}_+
$$
Using the cardinality from equation \eqref{indexed-n-card} and \eqref{def-n-union} with the disjoint property we do get:-
$$
|\mathcal{N} | = \sum_{k=2}^{\aleph_0} | \mathcal{N}(k) | = \sum_{k=2}^{\aleph_0} \aleph_1 M(k)
$$
Using continuum hypothesis (definition \ref{CH} ) the bounds for $M(k)$ becomes:
$$
\aleph_0 \le M(k) \le \aleph_1
$$
then, this would imply : 
$$
 \aleph_0 \aleph_1 \aleph_0 \le |\mathcal{N} | \le \aleph_0 \aleph_1 \aleph_1 
$$

which gets back immediately to:-
\begin{equation}\label{card-n}
|\mathcal{N}| = \aleph_1   
\end{equation} 

Hence, using equations \eqref{card-c} and \eqref{card-n} we get:-
$$
| \mathcal{C} | = | \mathcal{N} | = \aleph_1
$$   
which proves the theorem.
\end{proof}

That the set $\mathcal{S}$ has cardinality equal to the its own subsets, i.e. $\mathcal{C}$
and $\mathcal{N}$ should not surprise us. An example of this is the set of integers, 
which has the same cardinality $\aleph_0$ as that of the individual sets of even and odd integers.

The set, $\mathcal{N}$, is in some sense bigger than $\mathcal{C}$ even though 
they have same cardinality. Other examples exhibiting similar behavior include the set of continuous intervals within a line. 
The intervals  have same the cardinality ($\aleph_1$) 
as that of $\mathbb{R}^n$.
In effect the cardinality of the points in a line is the same as 
countably infinite dimensional Euclidean space. 

This brings us the question. What would be the chance that a randomly 
chosen sequence $s \in \mathcal{S}$ would belong to $\mathcal{C}$ or $\mathcal{N}$ ?
After all, the chance that a number, selected from a real ray $[0,\infty)$ would belong to 
any bounded line segment $[a,b]$ , is zero.
The idea is formalized next.  

\begin{definition}\label{sequence-space-function}
\textbf{Sequence to Space Mapping Function.}

Let us define $\mathcal{S}(n)$ to be,
\begin{equation}\label{s-n}
\mathcal{S}(n) = \bigcup_{k=1}^{n} \mathcal{N}(k) 
\end{equation}
and let:-
\begin{equation}\label{r-cauchy}
\mathbf{R}^n = [ r_k ] \; ; \; r_k \in [0,1] \; ; \; 0 < k \le n \; ; \;  i <j \implies r_i \le r_j  
\end{equation}
Then, the function :- 
\begin{equation}\label{f-s} 
f_{S} : \mathcal{S}(n) \rightarrow \mathbf{R}^n
\end{equation}
will be called Generalized Sequence to Space Mapping Function.
The co-ordinates are found as follows.

Let $lim( \mathcal{C}_k )$ denotes the limit of the Cauchy sequence $\mathcal{C}_k$.
Let $s \in \mathcal{S}(n)$ be a mixture of Cauchy sequences $|\{ \mathcal{C}_k \}| = m $. 
Then,
\begin{equation}\label{r-k}
r_k = \begin{cases}  lim( \mathcal{C}_k )  &\text{if }  0 < k \le m\\ 
 lim( \mathcal{C}_m )  & \text{if } m < k \le n 
\end{cases} 
\end{equation}

\end{definition}
In particular, if the sequence $s \in \mathcal{S}(n)$ is also in $\mathcal{C}$,
with $lim( \mathcal{C} ) = r $
then `$s$' will be mapped to $\mathbf{r}_C$ where,
\begin{equation}\label{cauchy-space}
\mathbf{r}_C = [ r_i ] \;  ; \; \forall i \;  r_i = r \; ;\; r \in [0,1]    
\end{equation}

We note that the curious $r \in [0,1]$ is the outcome of  
rationalization of the sequences  as definition \eqref{rat} with
all the Cauchy sequences having individual limit points, $c \in [0,1]$. 
If we have a sequence `$s_2$' of two Cauchy sequences mixed, 
individually converging to points $\{ a, b\}$, with $a\le b$, 
we would denote the sequence as $s_2 = ( a , b )$.
This ensures that the permutations of components does not matter.

\begin{definition}\label{sequence-space}
\textbf{Sequence Space.}
 
Every  sequence  $s \in \mathcal{S}(n)$ can be mapped into $\mathbf{R}^n$ (equation \ref{r-cauchy}) 
using the sequence space function $f_S$ (equation \ref{f-s}) of definition \eqref{sequence-space-function}.
The space so constructed is to be defined as sequence space.  
\end{definition}

\begin{definition}\label{sequence-measure}
\textbf{Sequence Measure.}

Let $A \subseteq \mathcal{S}(n)$ .
The sequence measure of $A$ is defined as,
$$
\mu(A) = n! V_n(A)
$$
where $V_n(A)$ is the volume occupied by the set $A \in \mathbb{R}^n$.  
\end{definition}

It is easy to check that the sequence measure (definition \ref{sequence-measure})
follows all the properties required for a measure (definition \ref{measure}).
Moreover, `$\mu$' is a probability measure, because,
$$
\mu(\mathcal{S}(n)) = n!V(\mathcal{S}(n)) = 1 \; ; \; \forall n > 0 .
$$

\begin{theorem}\label{null-sequence}
\textbf{ Measure of Cauchy Sequences is 0 over Sequence Space. }

Let $\mathcal{C} \subset \mathcal{S}(n)$. Then,
$$  
V_n(\mathcal{C})=0 \; ; \; n > 1 ,
$$ 
and subsequently $\mu(\mathcal{C}) = 0 $,
under the tenets of sequence measure (definition \ref{sequence-measure})
over sequence space (definition \ref{sequence-space}) whenever $n>1$.
\end{theorem}
\begin{proof}[Proof of theorem \ref{null-sequence}]

We note that under the effect of $f_S$ (equation \ref{f-s}) of definition \eqref{sequence-space-function} 
the Cauchy sequences are given by the equation \eqref{cauchy-space},
$$
\mathbf{r}_C = [ r_i ] \;  ; \; \forall i \;  r_i = r \; ;\; r \in [0,1]    
$$
represents a diagonal inside an `$n$'-dimensional hypercube for $n>1$.
For $n=1$ it is $\mathbb{R}$.
Then, for the case of $n=1$, the volume measure gets replaced with length measure
and we get :- 
$$
V_1(\mathcal{C}) = \lambda(\mathcal{C}) = 1 \text{ when } \mathcal{C} \subseteq \mathcal{S}(1)
$$ 
For $n=2$, it is the diagonal of a unit square, where the volume measure
has to be replaced by the area measure, and we get :-
$$
V_2(\mathcal{C}) = A(\mathcal{C}) = 0 \text{ when } \mathcal{C} \subset \mathcal{S}(2)  
$$  
For $n \ge 3$ it is the diagonal inside a unit hypercube, and hence:-
$$
V_n(\mathcal{C}) = 0 \text{ when } \mathcal{C} \subset \mathcal{S}(n) \; ; \; n \ge 3  
$$  
Hence, the set $\mathcal{C} \subset \mathcal{S}(n)$ with $n>1$ is a null set with 
sequence measure (definition \ref{sequence-measure}) (volume measure) 0.  
\end{proof}

With this discussion in mind, we formally present the theorem that states: 

\emph{Universal Computation, as debugger machine $\Delta_M$ runs it, 
is `almost surely' Chaotic. }

\begin{theorem}\label{ucc}
\textbf{Universal Computation performed by Debugger Machine is `Almost Surely' Chaotic. }

Let the machine $\Delta_M$ be given the task to simulate a random Turing Machine, $T_M$ ,
on a random input.
Then, the sequence $H = <h_n>$ of $\Delta_M$ will be chaotic with probability `1'.

\end{theorem}

\begin{proof}[Proof of theorem \ref{ucc}]
Under the tenet of $\Delta_M$ , $\mathcal{S}$ is the sample space. Hence,
$$
P(H \in \mathcal{S}) = 1
$$  
We note from equation \eqref{tp} that :-
\begin{equation}\label{tp_1}
P( H \in \mathcal{S} ) = P( H \in \mathcal{F} ) + P( H \in \mathcal{C} ) + P( H \in \mathcal{N} ) 
\end{equation}
The probability that $H$ will have a finite length is given by Lemma \ref{debugger-h-cardinality}.
\begin{equation}\label{tp_2}
P(Finite) = P( H \in \mathcal{F} ) = 0 .
\end{equation}
Hence, the $\Delta_M$ will almost surely not be of finite length.
This implies that the debugger machine \emph{`almost surely'} won't halt. Hence,
$$
P(H \in  \mathcal{C} \cup \mathcal{N} ) = 1 .
$$
We note here that the sequence measure (definition \ref{sequence-measure}) is a probability measure.
That is because in the general case, the measure adds up to the volume of a unit hypercube, 
as defined. In this regard, theorem \eqref{null-sequence} implies:-
$$
\forall n>1 \; ; \; V_n(\mathcal{C} \subset \mathcal{S}(n)) = 0  \implies P(s \in \mathcal{C}) = 0 
$$ 
for any $\mathcal{S}(n)$ with $n>1$.
However, we note that:- 
$$
\mathcal{S} = \lim_{n \to \aleph_0} \mathcal{S}(n).
$$ 
Applying Lemma \eqref{debugger-h-cardinality} and theorem \eqref{null-sequence}, 
for $H = <h_n>$, we get,
\begin{equation}\label{tp_3}
P(H \in \mathcal{C} ) = P(H \in \mathcal{N}(1) ) = 0 
\end{equation}
Clearly then using equations \eqref{tp_1} and \eqref{tp_2} and \eqref{tp_3} , 
the probability that the machine $\Delta_M$ does eventually end up in a Non-Cauchy sequence is,
\begin{equation}\label{tp_4}
P(H \in \mathcal{N} ) = 1 
\end{equation}
However, using theorem \eqref{ncabmc} we can say, when $H \in \mathcal{N}$ , that is the condition of chaos.
Hence, 
$$
P(<h_n>  \; \text{is  Chaotic}  ) = P(H \in \mathcal{N} ) = 1 ;  
$$ 
which establishes the theorem.
\end{proof}

\end{section}

\begin{section}{Implication}\label{implication}
Theorem \eqref{ucc} actually means that:-
 
\emph { A randomly drawn discrete state Universal Computing Mechanism, given a random input, 
would behave chaotically `almost surely'. }

The notion ``\emph{chaotically}''  is being defined as a condition when sequences of the tape symbols 
comprise multiple Cauchy sequences. Both the terms \emph{chaotic} and \emph{multiple Cauchy}
have deep implications, which is discussed further in the following subsections.

\begin{subsection}{Halting Problem : A New Perspective}
We start with the celebrated halting problem \eqref{halting-problem} itself.
Yes, it is a fact that we can \emph{`almost surety'} conclude that Turing Machines would not halt. 
We \emph{know now} that this is due to the presence of chaotic sequences!
We will demonstrate that in the below paragraph.

Let's assume for the moment that \emph{somehow} in some universe the sequences `\emph{almost surely}' do not get chaotic.
In that case,
$$
P(H) =  P(<h_n>  \; is \; Convergent \; ) = 1 
$$  
which would mean, the only problem is the precision of the result, not the result itself!
One can halt a computation just by being imprecise. The situation is not much different as of today.
Any attempt to find an irrational number with arbitrary precession  would not ever halt.
But then, with the choice of precision, we can halt it any time, with probability 1.

However, due to the  chaotic component (imprecise term) we can not do it for any input, or any sequence.
Thus, the crux of the halting problem 
stems from the undecidability  of \emph{`the fixed points of convergences'}. 
This in turn represents the chaotic part of computation.
\end{subsection}

\begin{subsection}{Simulating Neurons : Brain and Mind}
From time immemorial, the debate has always been around a mechanistic view of ``mind'' and the workings of the brain.
The mechanistic view applies Church Turing Thesis, upon brain, 
and treats it as just another computing device. 
This makes mind just as powerful as a Turing Machine.  
This view expresses the possibility of Turing Machines \emph{simulating} mind!

However, as we already established, Universal Computation can be chaotic. 
So even if brain works as a ``Universal Computer'', 
it might not be enough to predictably simulate it. 
The problems are two fold. 

\begin{enumerate}
\item{
The problem lies in applying the model of Turing machine. There is no provision to change data on the tape    
non deterministically in a Turing Machine. Real world behaves differently, 
where data changes non deterministically. 
Thought process, no matter how deep, eventually halts. 
Our perception changes based on data gathered by our senses. 
Mind is clearly not a \emph{``constant program, constant data''} computation. 
}
\item{
Another problem is due to the intrinsic property of any chaotic system,
i.e. \emph{sensitive dependence on initial conditions}. 
How can one set the initial conditions? With what precision ? 
Rational numbers do not have  finite precision. Clearly, even if we solve the problem of finding 
all the input parameters, a very small loss of precision would guarantee 
a totally different behavior of the simulated system under Turing Machine model.
}
\end{enumerate}
\end{subsection} 

\begin{subsection}{Learning Processes}
 We reflect about the learning process, where patterns are present
within large set of data set $X$, combined from different sources.
The learning process is all about discovering a function,  
$\mathcal{L}_k : X \to X $ such that  $\mathcal{K} \subseteq X $
can be recreated by using the fixed point iteration map $\mathcal{L}_k$ . 
If that is indeed the case, we may say that the system has learnt $\mathcal{L}_k$ .

But then there might be a set of such functions 
$\mathcal{L} = \{ \mathcal{L}_i \} $ operating on the domain $X$.
Generalized learning is all about discovering such 
$\mathcal{L}_k$ and adding them to the set of known $\mathcal{L}$.

A Randomly Drawn Turing Machine might be capable of producing the set $\mathcal{K}$ , alongside other sets.      
From the biological learning perspective, this has enormous value. 
Such a system is capable of detecting if the current state
was encountered earlier (e.g The debugger machine $\Delta_M$ ).
Further if it can somehow also detect the countable Cauchy subsequences 
$\mathcal{C}_i$, it can then  ``\emph{memorize}'' the process through which the sequences could be regenerated. 
This would be a ``\emph{learning}'' for the system in true sense.  

It can go further. The system which wants to concentrate on the subsequence $\mathcal{C}_r$ 
to arrive at the solution $\mathbf{a_r}$ , i.e.  $\mathcal{C}_r$'s limit point, 
can heuristically ``\emph{guess}''  and approximate the solution. 
This can happen even without solving the Turing Halting Problem as well.

The clustering of points in $X_t$ or the $H$ space, can be taken
as the signal that the Turing machine won't ever halt with the input.  
For more pathological cases, where the $H$ space would spread the whole of $X_t$ ,
more sophisticated guesses should be needed to heuristically stop and learn.

\end{subsection}

\end{section}

\begin{section}{Summary and Future Works}\label{summary}

In this paper, we have demonstrated that the Universal computation 
performed by the debugger machine \emph{is `almost surely' chaotic}.
The debugger machine is just a trick to formally showcase the inner chaotic nature of the universal computation.
The implication of this statement has far reaching consequences in the philosophy of learning as discussed in Section \ref{implication}. 

Historically, computation is thought to be all about finding directed answers by 
deterministic, mechanical, step by step procedures. In nature computation arose automatically, 
by the natural force of evolution, over billions of years, without any \emph{external programming}.
The solution to this puzzle lies in the idea that any mechanism capable of universal computation 
can be chaotic.

\emph{Chaotic, non halting computation} might generate  very good approximate solutions 
to multiple, unrelated problems.This is under the assumption 
that a computation can be halted at will. 
In nature, this happens frequently where a living organism simply dies!  
But those who survive, takes the \emph{learning} forward.

This paper is a small step towards the journey of solving the mystery of natural computation.
Further research would ensure the continuing quest for the most astonishing journey 
that the humans have ever embarked upon.
This is the journey to understand ourselves, who, in the immortal words of Carl Sagan are :
``unconscious matter grew into consciousness''.

\end{section}

\appendix
\begin{section}{Definitions of Section \ref{fpi-chaos}}\label{ap_1}
\begin{definition}\label{fp}
\textbf{Fixed Point of a function. }

For a function $f:X \to X$ , $x^*$ is said to be a fixed point, iff $f(x^*) = x^*$ .
\end{definition}

\begin{definition}\label{mp}
\textbf{Metric Space.}

A metric space is an ordered pair $(M,d)$ where $M$ is a set and $d$ is a metric on $M$ , i.e., a function:-

$$
d : M \times M \to \mathbb{R}
$$

such that for $x,y,z \in M$ , the following holds:-
\begin{enumerate}
\item { $d(x,y) \ge 0 $ }
\item { $ d(x,y) = 0 $ iff $x=y$ . }
\item { $d(x,y) = d(y,x) $ }
\item { $d(x,z) \le  d(x,y) + d (y,z)$ }
\end{enumerate}

The function `$d$' is also called ``distance function'' or simply ``distance''.
\end{definition}

\begin{definition}\label{cs}
\textbf{Cauchy Sequence in a Metric Space $(M,d)$ .}

Given a Metric space $(M,d)$ , the sequence $x_1,x_2,x_3,...$ of real numbers is called `Cauchy Sequence', if for every positive real number $\epsilon$ , 
there is a positive integer $N$ such that for all natural numbers $m,n > N$ the following holds:-
$$
d (x_m , x_n ) < \epsilon .
$$

\end{definition}
Roughly speaking, the terms of the sequence are getting closer and closer together in a way 
that suggests that the sequence ought to have a limit $x^*  \in M$ . Nonetheless, such a limit does not always exist within $M$ .

Note that by the term: \emph {sequence} we are implicitly assuming \emph {infinite sequence} , unless otherwise specified.
We will explicitly use finite sequences in Section \ref{chaos-computation}. 

\begin{definition}\label{cncms}
\textbf{Complete Metric Space.}

A metric space $(M,d)$ is called complete (or Cauchy) 
iff every Cauchy sequence (definition \ref{cs}) of points in $(M,d)$ has a limit , that is also in $M$ .
\end{definition} 

As an example of not-complete metric space take $\mathbb{Q}$ , the set of rational numbers. 
Consider for instance the sequence defined by $x_1 = 1$ and function $d$ is defined by standard
difference between $d(x,y) = |x-y|$ , then :-
$$
x_{n+1} = \frac{1}{2} \left ( x_n + \frac{2}{x_n} \right ) 
$$

This is a Cauchy sequence of rational numbers,
but it does not converge towards any rational limit, but to
$\sqrt{2}$ , but then $\sqrt{2} \not \in \mathbb{Q}$ . 

The closed interval $[0,1]$ is a Complete Metric space.

\begin{theorem}\label{bfpt}
\textbf{Banach Fixed Point Theorem \cite{bf} .}

Let $(X, d)$ be a non-empty complete metric space (definition \ref{cncms}) . 
Let $f : X \to X$ be a contraction mapping on $X$ , i.e.: there is a nonnegative real number $q < 1$  such that:-

\begin{equation}\label{cont_map}
d(f(x),f(y))  \le q . d(x,y) \; ; \; \forall x,y \in X  
\end{equation}

Then the map $f$ admits one and only one fixed point (definition \ref{fp} ) $x^*$ in X.  
Furthermore, this fixed point can be found as follows:- 
   
Start with an arbitrary element $x_0 \in X$ 
and define an iterative sequence by $x_n = f(x_{n-1})$ for $n = 1, 2, 3, ...$ .
This sequence converges, and its limit is $x^*$ . 
\end{theorem}

\begin{definition}\label{orbit}
\textbf{Orbit.}

Let $f:X \to X$ be a function. 
The sequence $ \mathcal{O} = \{x_0, x_1,x_2,x_3,...\}$ where
$$
x_{n+1} = f(x_n) \; ; \; x_n \in X \; ; \; n \ge 0 
$$
is called an orbit (more precisely `forward orbit') of $f$. 

$f$ is said to have a `closed' or `periodic' orbit $ \mathcal{O}$ if $| \mathcal{O}| \ne \infty$ .
\end{definition}

\begin{definition}\label{top-space}
\textbf{Topological Space.}

Let the empty set be written as : $\emptyset$. Let $2^X$ denotes the power set, i.e. the set of all subsets of $X$.
A topological space is a set $X$ together with $\tau \subseteq 2^X$ satisfying the following axioms:-
\begin{enumerate}
\item{ $\emptyset \in \tau$ and $X \in \tau$ ,}
\item{ $\tau$ is closed under arbitrary union, }
\item{ $\tau$ is closed under finite intersection. }
\end{enumerate}
The set $\tau$ is called a topology of $X$.
\end{definition}

\begin{definition}\label{set-cover}
\textbf{Cover (Topology).}

Let 
$C = \{ U_\alpha : \alpha \in A \}$ is an indexed family of sets $ U_\alpha$.

$C$ is a cover of set $X$ if :-
$$
X \subseteq \bigcup_{\alpha \in A} U_\alpha 
$$
\end{definition}

\begin{definition}\label{basis}
\textbf{Basis (Topology).}

A set of sets $B = \{ B_\alpha \}$ are said to be a basis of topology over $X$ if:-
\begin{enumerate}
\item{ $B$ covers ( definition \ref{set-cover} ) $X$. }
\item{ Let $B_1, B_2 \in B$. Let $x \in B_1 \cap B_2$. 
Then there is a basis element $B_3$ such that $x \in B_3$ and $B_3 \subset B_1 \cap B_2$.
}
\end{enumerate}

\end{definition}

\begin{definition}\label{dense-set}
\textbf{Dense Set.}

Let $A$ be a subset of a topological space $X$. 
$A$ is dense in $X$ for any point $x \in X$, if any neighborhood of $x$ contains at least one point from $A$.
\end{definition}

The real numbers $\mathbb{R}$ with the usual topology have the rational numbers $\mathbb{Q}$ as a countable dense subset.

\begin{definition}\label{top-trans}
\textbf{Topological Transitivity(Mixing).}

A function $f:X \to X$ is topologically transitive(Mixing) if, given any every pair of non empty open sets $A,B \subset X$ , 
there is some positive integer $n$ such that 
$$
f^n(A) \cap B \ne \emptyset .
$$ 
where $f^n$ means n'th iterate of $f$ .

\end{definition}

\begin{definition}\label{bounded-seq}
\textbf{ Bounded Sequence.}

A sequence $<x_n>$ is called a bounded sequence iff :-
$$
\forall n \; \; l \le x_n \le u \; ; \;  -\infty < l  \le  u < \infty  
$$

The number `l' is called the lower bound of the sequence and 
`u' is called the upper bound of the sequence.
\end{definition}

\begin{lemma}\label{bcss}
\textbf{Bolzano-Weierstrass.}

Every bounded sequence has a convergent (Cauchy) subsequence.
\end{lemma}

It is to be noted that a bounded sequence may have many convergent subsequences (for example, a sequence consisting of a counting of the rationals has subsequences converging to every real number) or rather few (for example a convergent sequence has all its subsequences having the same limit).

\begin{definition}\label{bncs}
\textbf{ Bounded Non-Cauchy Sequence in a metric space.}

A bounded sequence in a metric space is Non-Cauchy, iff the condition of definition \ref{cs} is not met.  
\end{definition}

\begin{lemma}\label{ncbs}
\textbf{Existence of countable Cauchy subsequences in a Bounded Non-Cauchy Sequence.}

A Bounded non-cauchy sequence in a metric space will contain countable number of Cauchy subsequences. 
\end{lemma}
\begin{proof}[Proof of the Lemma \ref{ncbs} ]

We proof this by construction. 
Let $\mathcal{N} = <x_n> $ be a non Cauchy bounded sequence.
 
By using lemma \ref{bcss} we must have one Cauchy subsequence on  $\mathcal{N}$, 
and $\mathcal{N}$ itself is not a Cauchy sequence (by assertion).
We can isolate the cauchy subsequence and call it $\mathcal{C}_1$. 
Now, if one removes the sequence $\mathcal{C}_1$ from $\mathcal{N}$ , 
we have a new sequence $\mathcal{N}_1$. 

However, this $\mathcal{N}_1$ is also a bounded sequence, as any subsequence of a bounded sequence is also bounded.
Then, repeating lemma \ref{bcss} again, we must have one Cauchy subsequence on  $\mathcal{N}_1$ , and we call it $\mathcal{C}_2$.
That $\mathcal{C}_1 \ne \mathcal{C}_2$ is certain, because, else both the sequences can be merged together.

We can repeat this procedure, till step $n$ when sequence $\mathcal{N}_n$ will not have any element left.
By that time, we have a set of Cauchy Subsequences:-
$$
\mathcal{C} = \{ \mathcal{C}_1, \mathcal{C}_2, \mathcal{C}_3 ,...,\mathcal{C}_n \} \; ; \; n \ge 2 ,
$$  

with the limit points (points of attraction) set:-

$$
\mathcal{A}=  \{ a_1, a_2, a_3,..., a_n \} ;
$$

with :-
$$
a_i = a_j \implies i = j . 
$$ 

That is, no  $\mathcal{C}_i , \mathcal{C}_j$ have the same limit, unless $i=j$ , by virtue of the construction. 

It is possible if $|\mathcal{N}| \to \infty$ that the $n \to \infty$ , but $n$ remains a natural number, and hence, 
the set $\mathcal{C}$ remains countable. 
The construction of the set `$\mathcal{C}$' proves the lemma.
\end{proof}
In informal language, any bounded Non-Cauchy sequence is a result of ``mixing'' countable
Cauchy sequences, each of them having their own limit points. 
That means, in case there is a function $f:X \to X$ , whose orbit generates such bounded Non-Cauchy  
sequence, `$f$' essentially then acts as a ``mixer'' of multiple (but countable) Cauchy Sequences.

\end{section}

\begin{section}{Definitions for Section \ref{chaos-computation}}\label{ap_2}

\begin{definition}\label{TM}
\textbf{Turing Machine.}

A ``Turing Machine'' is a 7-tuple ($Q,\Sigma,\Gamma,\delta,q_0,q_a,q_r$), where:-
\begin{enumerate}
\item{ $Q$ is the set of states. }
\item{ $\Sigma$ is the set of input alphabets not containing the blank symbol $\beta$. }
\item{ $\Gamma$ is the tape alphabet , where $\beta \in \Gamma$ and $\Sigma \subseteq \Gamma$. }
\item{ $\delta : Q \times \Gamma \to Q \times \Gamma \times \{L , R \} $ is the transition function. }
\item{ $q_0 \in Q$ is start state.}
\item{ $q_a \in Q$ is the accept state.}
\item{ $q_r \in Q$ is the reject state.}
\end{enumerate}
\end{definition} 
According to standard notion $q_a \ne q_r$ , but we omit this requirement here, 
as we are not going to distinguish between two different types of halting (`accept and halt' vs `reject and halt') of Turing Machines.

A Turing Machine `$M$' (definition \ref{TM}) computes as follows.

Initially `$M$' receives the input $w=w_1w_2w_3...w_n \in \Sigma^* $ 
on the leftmost `$n$' squares on the tape, and the rest of the tape is filled up with blank symbol `$\beta$'.
The head starts on the leftmost square on the tape. 
As the input alphabet `$\Sigma$' does not contain the blank symbol `$\beta$', 
the first `$\beta$' marks end of input.

Once `$M$' starts, the computation proceeds wording to the rules of `$\delta$'.
However, if the head is already at the leftmost position, then, 
even if the `$\delta$' rule says move `$L$' , the head stays there.

The computation continues until the current state of the Turing Machine is either $q_a$ , or $q_r$ .
In lieu of that, the machine will continue running forever.

\begin{definition}\label{decider}
\textbf{Decider Turing Machine.}

A Turing Machine, which is guaranteed to halt on any input (i.e. reach one of the states \{$q_a,q_r$\} ) is called a decider.
\end{definition}

\begin{definition}\label{undecidable}
\textbf{Undecidable Problem.}

If for a given problem, it is impossible to construct a  decider (definition \ref{decider}) Turing Machine, 
then the problem is called undecidable problem.
\end{definition}

\begin{definition}\label{UTM}
\textbf{Universal Turing Machine.}

An `UTM' or `Universal Turing Machine' is a Turing Machine (definition \ref{TM}) such that it can simulate an 
arbitrary Turing machine on arbitrary input.
\end{definition}

\begin{definition}\label{CTT}
\textbf{Church Turing Thesis.}

Every effective computation can be carried out by a Turing machine (definition \ref{TM}), 
and hence by an Universal Turing Machine(definition \ref{UTM}).
\end{definition}

\begin{definition}\label{halting-problem}
\textbf{Halting problem.}

Given a description of a computer program, decide whether the program finishes running or continues to run forever. 
This is equivalent to the problem of deciding, given a program and an input,
whether the program will eventually halt when run with that input, or will run forever.
\end{definition}

\begin{theorem}\label{halting-theorem}
\textbf{Halting Problem  is ``undecidable'' (Turing).}

There can not be any general decider (definition \ref{decider}), 
which can decide on the Halting Problem. In other words, 
the problem of definition \ref{halting-problem} is ``Undecidable'' (definition \ref{undecidable}).
\end{theorem}   

\end{section}

\begin{section}{Definitions of Section \ref{prob-chaos-tm}}\label{ap_3}

\begin{definition}\label{sigma-a}
\textbf{Sigma ( $\sigma$ ) Algebra. }

Let $X$ be some set, with the notation $2^X$ is it's power set. Then the subset $\Sigma \subset 2^X$ is a Sigma ($\sigma$) algebra iff:-
\begin{enumerate}

\item{$\Sigma$ is non empty.}
\item{$\Sigma$ is closed under complementation, if $A \subset \Sigma $ , then  $ X  \setminus A \subset \Sigma $.}
\item{$\Sigma$ is closed under countable union, if $A_1 , A_2 , A_3,... \subset \Sigma $ , then \\
 $ A_1\cup A_2\cup A_3,... \subset \Sigma $.}

\end{enumerate}
\end{definition}

\begin{definition}\label{measure}
\textbf{Measure. }

Let $\Sigma$ be a sigma algebra (definition \ref{sigma-a}) over a set $X$. A function $\mu : \Sigma \to \bar{\mathbf{R}}$,
where $\bar{\mathbf{R}} = [-\infty, +\infty]$ is extended real line; is called a measure iff it satisfies the following properties:-

\begin{enumerate}

\item{ Non Negativity : 
$$
\mu ( E ) \ge 0 \; ; \; \forall E \in \Sigma
$$
}

\item{ Countable Additivity : For all countable collections $\{E_i\}_{i \in I}$ of pairwise disjoint set in $\Sigma$
$$
\mu \left ( \bigcup _{i \in I} E_i \right ) = \sum _{i \in I} \mu(E_i)   
$$
}

\item {Null Empty Set:
$$
\mu ( \emptyset ) = 0 
$$

}

\end{enumerate}

\end{definition}

\begin{definition}\label{measure-s}
\textbf{Measure Space. }
A measure space is a measurable space, that is a set considered together with the sigma-algebra (definition \ref{sigma-a}) on the set,
possessing a non negative measure (definition \ref{measure} ).

\end{definition}

\begin{definition}\label{CH}
\textbf{Continuum Hypothesis.}

Let $|\mathbb{Z}| = \aleph_0$ and $|\mathbb{R}|=\aleph_1$.  
If S is a set, then,  
$$
\not \exists S \; s.t. \; \aleph_0 < |S| < \aleph_1.     
$$
\end{definition}

\begin{definition}\label{GCH}
\textbf{Generalized Continuum Hypothesis.}

Let A be a set with $|A|=\aleph_{\alpha}$ , with ordinal $\alpha$ and the power set of A is denoted by $2^A$ , then:-
 
$$
|2^A| = \aleph_{\alpha + 1} = {\aleph_{\alpha}}^{\aleph_{\alpha}}
$$

and there is no set S such that  $\aleph_{\alpha} < |S| <\aleph_{\alpha + 1}$.
The ordinal of $\mathbb{Z}$ is 0, and hence that of $\mathbb{R}$ is 1, in line with the continuum hypothesis (def \ref{CH}). 
\end{definition}

\begin{definition}\label{null-set}
\textbf{Null Set.}

Let $X$ be a measurable space, let $\mu$ be a measure on $X$ , and let $N$ be a measurable set in $X$. 
$N$ is null set(or zero measure), if its measure $\mu(N)=0$ .
\end{definition}

To be noted that, the set $N$ might not be empty set at all. In fact, the set $\mathbb{Z}$ is null, 
with respect to the measurable set $\mathbb{R}$, while  $|\mathbb{Z}| = \aleph_0$ , that is,
there are countably infinite number of elements in that set.

\begin{definition}\label{prob-m}
\textbf{Probability Axioms (Kolmogorov axioms). }

Let ( $\Omega, F, P$ ) be a measure space, with $P(\Omega) = 1$. Then ($\Omega, F, P$ ) is a probability space with 
sample space $\Omega$ , event space $F$ , and probability measure $P$.  
\end{definition}

\begin{definition}\label{as}
\textbf{Almost Surely or With Probability 1. }

Let ( $\Omega, F, P$ ) be a probability space. One says that an event $E \in F$ happens almost surely if $P(E) = 1$.
\end{definition}

\begin{definition}\label{an}
\textbf{Almost Never or With Probability 0. }

Let ( $\Omega, F, P$ ) be a probability space. One says that an event $E \in F$ happens almost never if $P(E) = 0$.
\end{definition}

Let $S \in F$ be an \emph{Almost Sure} event, that is $P(S)=1$ . 
Clearly, the complimentary event $F \setminus S$ is an \emph{Almost Never} event with  $P(F \setminus S)=0$ .

We now present  theorems to showcase these concepts.

\begin{theorem}\label{non-null-select}
\textbf{ Selection probability of an element from a null set is 0.}

Let $N$ be a null set such that $N \subset X$ , and $X$ is measurable.
If an element $x$ is randomly selected from the set $X$, then:-
$$
P(x \in N) = 0 .
$$ 
In other words $P(x \not \in N) = 1$.
\end{theorem}
\begin{proof}[Proof of theorem \ref{non-null-select}]

We note that $P$ is a measure over $X$ , and clearly if $N$ is a null set, then $P(N)=0$ by definition.
But $P(N)$ is the same as saying $P(N)= P(x \in N)$. Hence $P(x \in N)=0$.
But we know that:-
$$
P(X) = P(x \not \in N) + P(x \in N) .
$$
Therefore, by rearranging:-
$$
P(x \not \in N) = P(X) - P(x \in N) . 
$$
But then $P(X)=1$ and $P(x \in N)=0$.
Hence, $P(x \not \in N) = 1$.
\end{proof}
The result of theorem \ref{non-null-select} can be succinctly stated as:-

``That a randomly drawn number from real line $\mathbb{R}$ 
is \emph{`almost surely'} irrational''.

\end{section}

\end{document}